\documentclass{article}

\usepackage{microtype}
\usepackage{graphicx}
\usepackage{subfigure}
\usepackage{booktabs} %
\usepackage{multirow, multicol}
\usepackage[table]{xcolor} 
\usepackage{comment}
\usepackage{makecell}
\usepackage{float}
\usepackage{tabularx}
\usepackage[textsize=tiny]{todonotes}
\usepackage{hyperref}

\usepackage[accepted]{icml2025}

\usepackage{amsmath}
\usepackage{amssymb}
\usepackage{mathtools}
\usepackage{amsthm}

\usepackage[capitalize,noabbrev]{cleveref}

\theoremstyle{plain}
\newtheorem{theorem}{Theorem}[section]
\newtheorem{proposition}[theorem]{Proposition}

\theoremstyle{definition}

\theoremstyle{remark}

\newcommand{\norm}[2][]{\left\lVert #2 \right\rVert_{#1}}

\usepackage[textsize=tiny]{todonotes}

\icmltitlerunning{Straight but not so fast: Challenges with Rectified Flows in Protein Design.}

\begin{document}

\twocolumn[
\icmltitle{Flows, straight but not so fast: \\Exploring the design space of Rectified Flows in Protein Design.}

\begin{icmlauthorlist}
\icmlauthor{Junhua Chen}{1}
\icmlauthor{Simon V. Mathis}{1}
\icmlauthor{Charles Harris}{1}
\icmlauthor{Kieran Didi}{2,3}
\icmlauthor{Pietro Lio}{1}
\end{icmlauthorlist}

\icmlaffiliation{1}{University of Cambridge, UK}
\icmlaffiliation{2}{University of Oxford, UK}
\icmlaffiliation{3}{NVIDIA}

\icmlcorrespondingauthor{Junhua Chen}{junhua.chen.ios@gmail.com}

\icmlkeywords{Machine Learning, ICML}

\vskip 0.3in
]

\printAffiliationsAndNotice{}  %

\begin{abstract}
Generative modeling techniques such as Diffusion and Flow Matching have achieved significant successes in generating designable and diverse protein backbones. However, many current models are computationally expensive, requiring hundreds or even thousands of function evaluations (NFEs) to yield samples of acceptable quality, which can become a bottleneck in practical design campaigns that often generate $10^4-10^6$ designs per target. In image generation, Rectified Flows (ReFlow) can significantly reduce the required NFEs for a given target quality, but their application in protein backbone generation has been less studied. We apply ReFlow to improve the low NFE performance of pretrained $SE(3)^N$ flow matching models for protein backbone generation and systematically study ReFlow design choices in the context of protein generation in data curation, training and inference time settings. In particular, we (1) show that ReFlow in the protein domain is particularly sensitive to the choice of coupling generation and annealing, (2) demonstrate how useful design choices for ReFlow in the image domain do not directly translate to better performance on proteins, and (3) make improvements to ReFlow methodology for proteins.
\end{abstract} %

\section{Introduction}
\label{section: Intro}

\begin{table*}[h]
\centering
\small
\renewcommand{\arraystretch}{1.3}
\begin{tabular}{@{}p{0.25\linewidth}p{0.30\linewidth}p{0.35\linewidth}@{}}
\toprule
\textbf{ReFlow Design Component} & \textbf{Standard Practices in Computer Vision} & \textbf{Protein-Specific Adaptations} \\
\midrule
\multicolumn{3}{@{}l@{}}{\textit{\textbf{Data Curation}}} \\[0.2em]
\quad Coupling Strategy & Forward coupling, \colorbox{red!20}{Inverted coupling} & Forward coupling \\[0.3em]
\quad \colorbox{yellow!20}{Inference Annealing} & Not applicable & \textbf{Model-dependent optimization} \\[0.5em]

\multicolumn{3}{@{}l@{}}{\textit{\textbf{Training Configuration}}} \\[0.2em]
\quad Time Sampling & Uniform, Exponential, Cosh & Standard distributions \\[0.3em]
\quad Loss Function & Flow matching loss only & \colorbox{green!20}{\textbf{+ Structural losses}}, Axis-angle loss \\[0.5em]

\multicolumn{3}{@{}l@{}}{\textit{\textbf{Inference Settings}}} \\[0.2em]
\quad Discretization Scheme & Uniform, Sigmoid, Exponential & \colorbox{green!20}{\textbf{Noise-focused discretization}} \\[0.3em]
\quad \colorbox{yellow!20}{Inference Annealing} & Not applicable & \textbf{Model-dependent optimization} \\
\bottomrule
\end{tabular}
\caption{\textbf{ReFlow design choices for protein generation reveal domain-specific requirements.} Comparison of standard ReFlow practices with protein-specific adaptations shows that \colorbox{yellow!20}{protein-unique} design choices significantly impact performance, while \colorbox{red!20}{suboptimal} approaches should be avoided and \colorbox{green!20}{beneficial} modifications substantially improve results.}
\label{tab:reflow_design_comparison}
\end{table*}

Generative modelling methods have become an important part of protein design workflows \cite{RFDiffusion, Chroma2023, bose2024se3stochasticflowmatchingprotein} to produce new proteins that satisfy prescribed structural and functional requirements. Increasingly, this setting of \textit{de novo} protein design is now seen as a promising direction for drug discovery and tackling major challenges in medicine.

However, during a practical protein design campaign, such generative models are often used to generate $10^4 \ - 10^6$ backbone samples (depending on the difficulty of the objective) \cite{lauko2025computational,Ahern2025.04.09.648075}, each sample often requiring hundreds of function calls due to the sequential nature of current approaches such as Flow Matching \cite{lipman2023flowmatchinggenerativemodeling} and Diffusion Models \cite{ho2020denoisingdiffusionprobabilisticmodels}. As a result, the high computational cost of sampling current protein generative models restricts the number of designs that can be explored for a given computational budget.

In computer vision, the problem of accelerating inference in flow matching models has been tackled with significant success using the Rectified Flow Algorithm (ReFlow) \cite{liu2022flowstraightfastlearning}. ReFlow provably leads to flow trajectories that are straighter, and consequently amenable to integration by relatively fewer steps. This reduction in integration steps accelerates generation, and has been shown to reduce image generation times by an order of magnitude while preserving image quality and distributional metrics \cite{kim2024simplereflowimprovedtechniques}. 

It is therefore natural to ask whether and how these techniques can be brought to protein design to improve the resource efficiency of design campaigns. Our paper addresses this topic. Our contributions are as follows:
\begin{enumerate}
\vspace{-2em}
\setlength{\itemsep}{0pt}
\setlength{\parskip}{0pt}
\item We generalize ReFlow to manifold data and apply it to pre-trained frame-based models for protein backbone generation.
\item We study the design space for applying ReFlow to protein backbone generation models by considering coupling generation, training and inference separately and give suggestions for design choices.
\item  We find that several design choices such as choice of inference setting (annealing, generation direction) in coupling generation and structural losses not present in other applications of ReFlow have large impacts on the performance of the rectified model, and give suggestions for these design choices. We also find that many proposed improvements used in ReFlow in other domains do not carry over to proteins and highlight where the problems lie.
\item We give guidelines of when rectification is worthwhile compared to simpler finetuning methods.
\end{enumerate}

\section{Background and related work}
Generative models for protein backbone design are a central element of the modern protein design workflow. Diffusion and Flow based models in particular are the methods of choice for backbone generation \cite{RFDiffusion,Ahern2025.04.09.648075}. 

\paragraph{Riemannian Flow matching}

Flow matching (FM) \cite{lipman2023flowmatchinggenerativemodeling, liu2022flowstraightfastlearning, albergo2023stochasticinterpolantsunifyingframework} is a method
for learning continuous normalizing flows (CNFs) for generative modelling, enabling samples from one distribution to be transported to another by integrating an ordinary differential equation (ODE) over a learned vector field. It has been extended to general Riemannian manifolds \cite{chen2024flowmatchinggeneralgeometries}, and is used as the theoretical basis for many protein backbone generation models \cite{yim2023fastproteinbackbonegeneration, bose2024se3stochasticflowmatchingprotein}. On a manifold $\mathcal{M}$, the CNF $\phi_t(\cdot)$ is defined by integrating along a time-dependent vector
field $v(x,t)\in \mathcal{T}_x\mathcal{M}$ where $\mathcal{T}_x\mathcal{M}$ is the tangent space of the manifold at $x\in \mathcal{M}$
\begin{equation}
    \frac{d\phi_t}{dt} = v(\phi_t(x),t),\ \phi_0(x)=x
\end{equation}
where we take $0\leq t\leq 1$. The objective is to learn $v$ such that if $X_0$ is sampled according to $p_0$, then $X_1=\phi_1(X_0)$, obtained by integrating the flow is distributed according to the target distribution $p_1$. We define the distribution of $X_t=\phi_t(X_0)$ to be $p_t$ and call $(p_t)_{0\leq t\leq 1}$ the probability path of the flow. The key insight of Flow Matching and Riemannian Flow Matching is that $v$ can be learned using a simulation free loss 
\begin{equation}
    \mathcal{L}_{FM}=\mathbb{E}[\|v_t(x)-u_t(x|x_1)\|^2_g]
\end{equation}
where $t\sim U[0,1]$, $\|\cdot \|_g$ is the norm induced by the Riemannian metric, $x_1\sim p_1$ is a sample from the target distribution, $x\sim p_t(\cdot|x_1)$ and where $p_t(\cdot|x_1)$ and $u_t(x|x_1)$ are corresponding \textit{conditional} probability path and vector fields which the target flow decomposes into. 

\paragraph{ReFlow algorithm}
Let $p_0, p_1$ be two data distributions on $\mathbb{R}^d$. Rectified Flow (RF) \cite{liu2022flowstraightfastlearning}
is an algorithm which learns ordinary differential equations (ODEs) transporting $p_0$ to $p_1$ and iteratively refines the drift such that the trajectories of the ODE become straight lines (and such flows are called \textit{straight}). ReFlow achieves this by first training a CNF similar to flow matching with straight line interpolant, and then repeatedly applies the \textit{rectification procedure} to it, by using the pre-trained model to generate a dataset of noise-data pairs (termed a \textit{coupling}) and then finetuning the model on this coupling to obtain a new model. The \textit{rectification procedure}, detailed in Algorithm \ref{alg:ReFlow} provably leads to straight flows in the infinite iteration limit.

\paragraph{Protein Backbone generative models}
Chemically, a protein is a chain of linked \textit{amino acids} (also referred to as residues) that folds under electrostatic forces into some 3D structure. The aim of protein backbone generative models is to generate plausible backbone structures which can be realized by some protein. A popular class of backbone generative models which will be the focus of this paper featurizes a protein structure as a sequence of rigid bodies, one per amino acid. Rigid bodies are in turn represented as frames, formally elements of $SE(3)$, and as such the whole protein is represented as an element of $SE(3)^N$. Models belonging to this class include RFDiffusion 1 and 2 \cite{RFDiffusion, Ahern2025.04.09.648075}, FrameDiff \cite{yim2023se}, FrameFlow \cite{yim2023fastproteinbackbonegeneration} and FoldFlow \cite{bose2024se3stochasticflowmatchingprotein}, and QFlow \cite{yue2025reqflowrectifiedquaternionflow}. 

\paragraph{Related Work:} Concurrently to our work on ReFlow for frame-based backbone generation models, ReQFlow \cite{yue2025reqflowrectifiedquaternionflow} implements ReFlow for a closely related quaternion-based formulation. In contrast to \citet{yue2025reqflowrectifiedquaternionflow}, our work focuses on studying the broader design space of ReFlow as applied to proteins, with the goal of understanding how the idiosyncrasies of protein structures as a data modality interacts with ReFlow. In particular, we systematically study how the common practice of inference-time annealing influences the ReFlow coupling and by consequence the rectified model, as well as show that small changes in ReFlow training can significantly affect final outcomes. We provide comparisons to ReQFlow throughout the paper to demonstrate how our insights extend to their work.

\section{ReFlow on $SE(3)^N$}

A key observation is that the ReFlow algorithm as presented in Algorithm \ref{alg:ReFlow} is almost directly applicable to manifold data. By replacing the Euclidean norm in $\mathcal{L}$ with the Riemannian metric and the straight line interpolant with the \textit{geodesic} interpolant, we obtain a version of Rectified Flows for data on manifolds. Whereas the Euclidean ReFlow Algorithm was shown to preserve marginal distributions in \citet{liu2022flowstraightfastlearning}, Theorem 2 of \citet{wu2025riemannianneuralgeodesicinterpolant} also shows that the rectification procedure also preserves marginal distributions of $X_t$ on data residing on a wide class of manifolds that includes the $SE(3)^N$ case which we are interested in. Moreover, Proposition \ref{prop:reudcetransport costs} shows that the property of ReFlow in $\mathbb{R}^N$ of reducing transport distances also carry over when ReFlow is performed on manifold data, thereby suggesting that ReFlow can also help improve flow matching couplings on manifold data. %

\begin{proposition}\label{prop:reudcetransport costs}
Let $(X_0, X_1)$ be the coupling used to train the rectified flow and $(Z_0, Z_1)$ be the coupling induced by the \textit{rectified} model. Then under the assumptions of Theorem 2 in \citet{wu2025riemannianneuralgeodesicinterpolant} we have\begin{equation}\label{eq:transportcost} 
    \mathbb{E}[d_g(Z_0, Z_1)]\leq \mathbb{E}[d_g(X_0,X_1)]
\end{equation} where $d_g$ is the geodesic distance induced by some Riemannian metric on the manifold.
\end{proposition}
\begin{proof} 
Same idea as Theorem D.5 of \citet{liu2022flowstraightfastlearning}. See Appendix \ref{section:proof}
\end{proof}

These results highlight that ReFlow is an viable algorithm in the frame-based protein model setting, and in subsequent sections we will study how best to apply ReFlow to proteins. 

\begin{algorithm}[]
   \caption{Rectification Iteration (ReFlow Algorithm)}
   \label{alg:ReFlow}
\begin{algorithmic}
   \STATE {\bfseries Input:} A trained flow velocity model $v_\theta(\cdot, t)$.
   \STATE {\bfseries Output:} The \emph{rectified} flow $v_{\theta'}(\cdot, t)$.
   \STATE
   \STATE {\bfseries Step 1: Generate Coupling}
   \STATE Sample pairs $(X_0, X_1)$ following $dX_t = v_\theta(X_t, t)\,dt$ 
   \STATE by starting from $X_0 \sim p_0$ or $X_1 \sim p_1$.
   \STATE
   \STATE {\bfseries Step 2: Rectification}
   \STATE Minimize $\theta$ in
   \[
   \mathcal{L}(\theta) = \mathbb{E}_{t, X_0, X_1}\left[\|v_\theta(X_t, t) - \partial_t I(X_0, X_1, t)\|^2\right]
   \]
   \STATE where $I(x, y, t)$ is the geodesic (straight line) interpolant between $X_0$ and $X_1$ to obtain $v_{\theta'}$ the \textit{rectified} flow.
\end{algorithmic}
\end{algorithm}

\section{Examining the design space of ReFlow}

We consider the design space of applying ReFlow to protein backbone generation. As shown in Table \ref{tab:reflow_design_comparison}, the ReFlow procedure can be partitioned into 3 phases: Data Generation, Training and Inference. In Section \ref{section:DataCuration}, we examine how the performance of a rectified model changes as we vary the method used to generate its training data. In Section \ref{section:improvingreflow} , we examine how the training methodology of the ReFlow algorithm can impact model performance for proteins. Lastly, we discuss Inference-time design choices in \ref{section:discretization}. We re-use the evaluation protocol of \citet{bose2024se3stochasticflowmatchingprotein}, reporting diversity, designability, self-consistency RMSD (scRMSD) and novelty (to the FoldFlow-OT's training set) as well as the number of designable foldseek clusters c.f \citet{geffner2025proteinascalingflowbasedprotein} when relevant. We focus on the FoldFlow-OT model \cite{bose2024se3stochasticflowmatchingprotein}  as well as the recent QFlow models \cite{yue2025reqflowrectifiedquaternionflow}. We defer most experimental details, as well as details of evaluation to Appendix \ref{section:experimentals}.

\subsection{Data Curation}\label{section:DataCuration}

\subsubsection{Experimental Setting}
Unlike for images, flow matching for proteins comes with a variety of inference-time settings that enables trading-off the diversity of the distribution (and its fidelity to the original training set) for sample quality \cite{geffner2025proteinascalingflowbasedprotein}. For instance, a particularly important setting for frame-based models studied in this section is the use of choice of \textit{inference annealing} for the $SO(3)$ velocity field  \cite{bose2024se3stochasticflowmatchingprotein}, where the rotational component of the predicted velocity field is scaled up during inference time, resulting in the rotational information of protein samples denoising faster than the translational component. Similar train/inference time mismatches also exist for non-frame based methods \cite{geffner2025proteinascalingflowbasedprotein}.

When curating the coupling used for ReFlow training, choices must be made for these parameters. As such, we study the effect that using different inference settings to generate the coupling has on the designability, diversity and secondary structure of the resulting rectified model. For FoldFlow, we generate a paired dataset of 25'100 protein backbone ($X_1$) - noise ($X_0$) samples,  comprised of 100 examples each length between 50-300 amino acids.  We adopt a straightforward ReFlow setup based on FoldFlow's training code, substituting the PDB training dataset for the generated coupling, with each training batch containing different proteins of the same length. For ReqFlow, we follow the curation process of the paper, except that we do not perform designability filtering on the data. When evaluating the model, we use the uniform discretization on $[0,1]$ for ReqFlow \cite{yue2025reqflowrectifiedquaternionflow} as in the original paper, but for FoldFlow we use a custom discretization in the 15 NFE setting as we found it was essential for both FoldFlow-OT and the rectified model to obtain reasonable designability scores. In all cases, we apply rectification once as in \citet{kim2024simplereflowimprovedtechniques}. We discuss discretizations further in Section \ref{section:discretization}.

\subsubsection{Results}\label{section:dataresults}
We present experimental results for FoldFlow-OT in Table \ref{tab:dataimpact} and Figure \ref{fig:secondary-structure-labels}, with corresponding QFlow results in Table \ref{tab:dataimpactreqflow} and Figure \ref{fig:secondary structure QFlow}.
ReFlow consistently improves low-NFE designability across all three data settings, consistent with prior image generation studies \citep{kim2024simplereflowimprovedtechniques, liu2022flowstraightfastlearning}, though sometimes at the expense of high-NFE performance. However, ReFlow exhibits strong sensitivity to the fine-tuning dataset choice. Rectified models tend to adjust to the statistical properties of their coupling datasets in terms of designability, diversity, and secondary structure characteristics, with higher inference annealing corresponding to higher designability but worse diversity.

This pattern is evident in both model architectures: QFlow models rectified on unannealed samples generate more diverse protein backbones with broader secondary structure support but reduced designability, while those trained on annealed samples show the opposite trend. Although distribution shifts could potentially result from memorization \citep{kim2024simplereflowimprovedtechniques, zhu2025analyzingmitigatingmodelcollapse}, our relatively small number of fine-tuning steps and the absence of overfitting (Table \ref{tab:nomemorization}) suggest that these shifts represent inherent byproducts of ReFlow in the protein domain rather than memorization artifacts.

These findings underscore the importance of selecting appropriate inference settings that achieve acceptable designability-diversity tradeoffs before applying ReFlow, because the coupling dataset choice directly influences the rectified model's distribution. This is exemplified by the intermediate inference scaling parameter (c=3) yielding the highest number of designable clusters in both base and rectified models.

\begin{figure}
    \centering
    \includegraphics[width=1.0\linewidth]{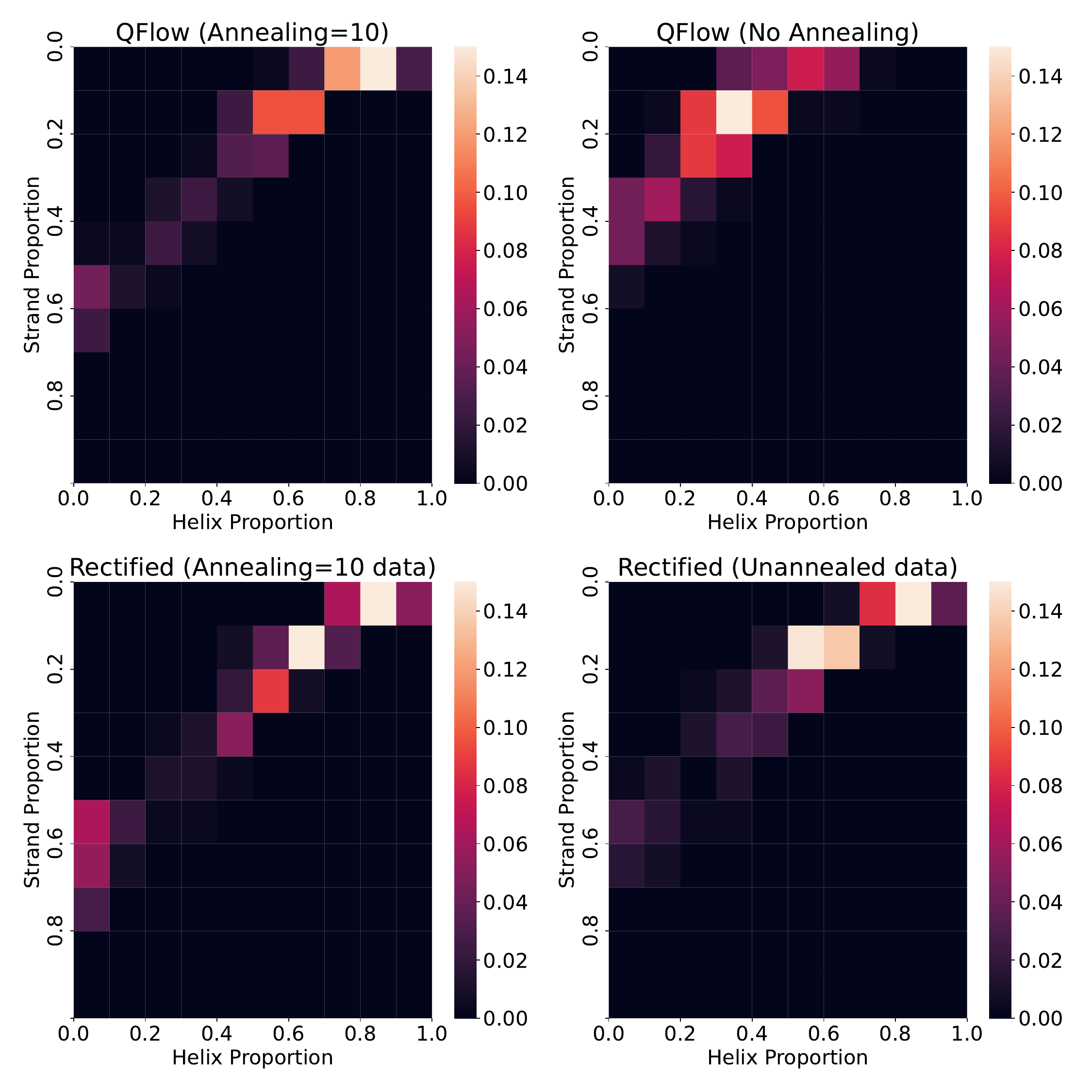}
    \caption{\textbf{Inference annealing settings affect secondary structure diversity through ReFlow coupling selection.} Secondary structure distributions for QFlow show that models rectified on unannealed samples exhibit wider support and improved diversity (Table \ref{tab:dataimpactreqflow}), demonstrating how ReFlow coupling choice can shift the modeled distribution.}
    \label{fig:secondary structure QFlow}
\end{figure}
\begin{table}[ht]
\centering
\resizebox{0.5\textwidth}{!}{%
\begin{tabular}{cc}
\toprule
\textbf{Model} & \textbf{Avg. Max-TM vs Finetune Set $\downarrow$}\\ 
\midrule
FoldFlow-OT  & $0.904\scriptstyle{\pm 0.006}$ \\
FoldFlow-OT-ReFlowed & $0.900\scriptstyle{\pm 0.005}$ \\
\bottomrule
\end{tabular}
}
\caption{\textbf{No signs of overfitting to the fine-tune dastaset}: The samples of FoldFlow-OT (FF-OT) and the rectified model (ReFF-OT) are compared against the rectification dataset. In this case the amount of fine-tuning iterations is insufficient to cause overfitting to the ReFlow coupling.}
\label{tab:nomemorization}
\end{table}

\begin{table*}[t]
\centering
\scriptsize
\begin{tabular}{ll|cc|cccc} \toprule \multirow{2}{*}{Dataset} & \multirow{2}{*}{Metric} & \multicolumn{2}{c|}{Data Generator (FoldFlow-OT)} & \multicolumn{4}{c}{ReFlowed Model (ReFoldFlow-OT)} \\ & & NFE=15, Yes & NFE=50, Yes & NFE=15, No & NFE=15, Yes & NFE=50, No & NFE=50, Yes \\ 
    \midrule 
    \multirow{5}{*}{Annealing=10} & Designability $\uparrow$ & $0.584 \scriptstyle { \pm 0.053 }$ & $\underline{\textbf{0.816}} \scriptstyle \pm 0.041$ & $0.564 \scriptstyle \pm 0.048$ & $0.664 \scriptstyle \pm 0.046$ & $0.656 \scriptstyle \pm 0.050$ & $\textbf{0.776} \scriptstyle \pm 0.044$ \\
    & Avg. scRMSD (\AA) $\downarrow$ & $3.088 \scriptstyle { \pm 0.376 }$ & $\underline{\textbf{2.120}} \scriptstyle \pm 0.315$ & $3.711 \scriptstyle \pm 0.450$ & $3.156 \scriptstyle \pm 0.381$ & $3.099 \scriptstyle \pm 0.374$ & $\textbf{2.313} \scriptstyle \pm 0.364$ \\ 
    & Diversity (Average-TM) $\downarrow$ & $0.425$ & $\underline{\textbf{0.397}}$ & $0.425$ & $0.442$ & $\textbf{0.411}$ & $0.430$ \\ 
    & \# Designable Clusters $\uparrow$ & $51$ & $\underline{\textbf{96}}$ & $60$ & $62$ & $\textbf{77}$ & $76$ \\ 
    & Novelty (Max-TM) $\downarrow$ & $0.819$ & $\textbf{0.815}$ & $\textbf{0.815}$ & $0.822$ & $\underline{\textbf{0.802}}$ & $0.816$ \\ \midrule 
    \multirow{5}{*}{Annealing=3} & Designability $\uparrow$ & $0.372 \scriptstyle { \pm 0.054 }$ & $0.472 \scriptstyle \pm 0.051$ & $0.188 \scriptstyle \pm 0.042$ & $\textbf{0.492} \scriptstyle \pm 0.049$ & $0.204 \scriptstyle \pm 0.040$ & $\underline{\textbf{0.552}} \scriptstyle \pm 0.053$ \\
    & Avg. scRMSD (\AA) $\downarrow$ & $5.716 \scriptstyle { \pm 0.534 }$ & $4.663 \scriptstyle \pm 0.469$ & $7.717 \scriptstyle \pm 0.470$ & $\textbf{4.432} \scriptstyle \pm 0.445$ & $7.429 \scriptstyle \pm 0.504$ & $\underline{\textbf{3.763}} \scriptstyle \pm 0.439$ \\ 
    & Diversity (Average-TM) $\downarrow$ & $0.362$& $0.345$ & $\textbf{0.210}$ & $0.379$ & $\underline{\textbf{0.205}}$ & $0.376$ \\
    & \# Designable Clusters $\uparrow$ & $78$ & $\underline{\textbf{99}}$ & $45$ & $80$ & $45$ & $\textbf{86}$ \\
    & Novelty (Max-TM) $\downarrow$ & $0.769$ & $0.769$ & $\underline{\textbf{0.749}}$ & $0.770$ & $\textbf{0.759}$ & $0.772$ \\ \midrule 
    \multirow{5}{*}{No Annealing} & Designability $\uparrow$ & $0.032 \scriptstyle { \pm 0.021 }$ & $0.044 \scriptstyle \pm 0.024$ & $0.028 \scriptstyle \pm 0.019$ & $\textbf{0.284} \scriptstyle \pm 0.047$ & $0.036 \scriptstyle \pm 0.021$ & $\underline{\textbf{0.320}} \scriptstyle \pm 0.046$ \\
    & Avg. scRMSD (\AA) $\downarrow$ & $10.977 \scriptstyle { \pm 0.430 }$ & $9.832 \scriptstyle \pm 0.453$ & $11.315 \scriptstyle \pm 0.408$ & $\textbf{6.471} \scriptstyle \pm 0.522$ & $11.570 \scriptstyle \pm 0.433$ & $\underline{\textbf{5.838}} \scriptstyle \pm 0.497$ \\
    & Diversity (Average-TM) $\downarrow$ & $\textbf{0.067}$ & $0.171$ & $\textbf{0.067}$ & $0.298$ & $\underline{\textbf{0.066}}$ & $0.264$ \\ 
    & \# Designable Clusters $\uparrow$ & $8$ & $10$ & $7$ & $\textbf{58}$ & $9$ & $\underline{\textbf{68}}$ \\ 
    & Novelty (Max-TM) $\downarrow$ & $0.750$ & $0.727$ & $\underline{\textbf{0.694}}$ & $0.756$ & $\textbf{0.703}$ & $0.745$ \\
    \bottomrule 
\end{tabular}

\caption{\textbf{Comparison of Rectified models under varying annealing schedules, function evaluations (NFE) and choices of training coupling for FoldFlow. “Yes/No” indicates whether inference-time annealing was used (if Yes, a value of 10 was used for annealing).} We also report the corresponding performance of the (model, inference setting)  pair used to generate each training coupling to give a reference. Several observations can be made: The use of inference annealing yields more designable proteins at the expense of reduced diversity for both the base and rectified models, with intermediate values of the annealing parameter being optimal with regards to the total number of unique designable clusters for the base model. The rectified model is strongly affected by the choice of training coupling, with the diversity and designability of the rectified model correlating strongly with the method used to generate the fine-tuning coupling, again with intermediate values of annealing being optimal for coupling generation. In line with prior works on ReFlow, rectification causes degradation in high NFE performance but significantly increases low NFE performance compared to the base model. Any reduction in resulting model designability can be offset by increased model throughput.}
\label{tab:dataimpact}
\end{table*}

\begin{figure*}[!t]
\centering
    \includegraphics[width=1.0\linewidth]{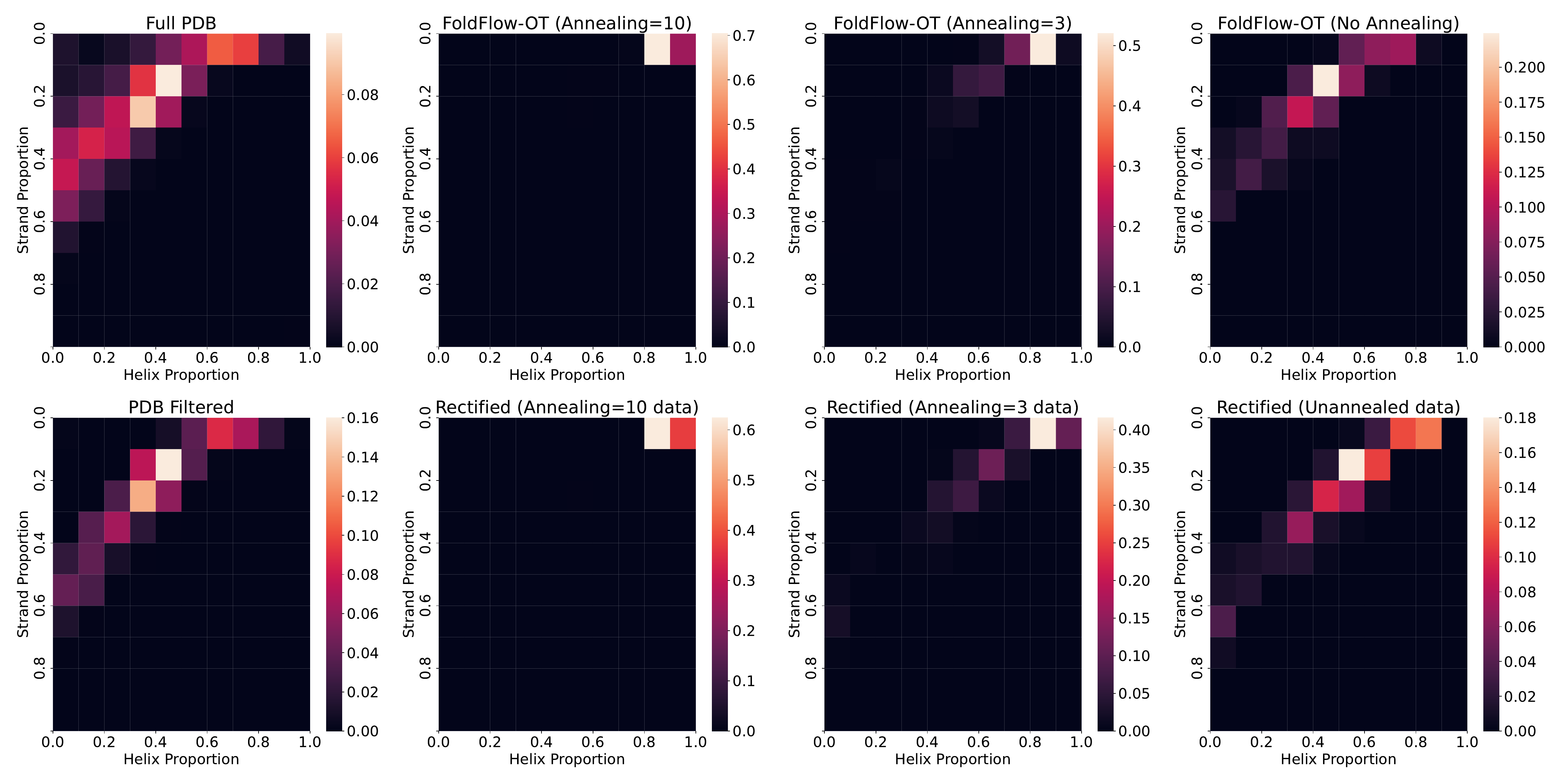}
    \vspace{-20pt}
    \caption{\textbf{Secondary structure distribution for FoldFlow-OT under different inference annealing settings (top) and that of models rectified on the corresponding coupling (bottom)}. The Rectified models were sampled with Inference Annealing=10. The secondary structure statistics of the rectified model is also strongly influenced by the fine-tuning coupling, with greater diversity in the fine-tuning distribution also translating to greater diversity in the fine-tuned model, although there is still some bias towards helical structures from the base model (FoldFlow-OT). This highlights the sensitivity of protein models to the fine-tuning distribution, simultaneously underlining the need to exercise caution when choosing a coupling to apply ReFlow on that is representative of the desired protein distribution, as well as potential opportunities in fine-tuning protein flow matching models on curated data. }
    \label{fig:secondary-structure-labels}
    \vspace{-1.5em}
\end{figure*}

\begin{table*}[t]
\centering
\scriptsize

\begin{tabular}{ll|cc|cccc} \toprule \multirow{2}{*}{Dataset} & \multirow{2}{*}{Metric} & \multicolumn{2}{c|}{Data Generator (QFlow)} & \multicolumn{4}{c}{ReFlowed Model (ReQFlow)} \\ & & NFE=20, Yes & NFE=100, Yes & NFE=20, No & NFE=20, Yes & NFE=100, No & NFE=100, Yes \\ \midrule
    \multirow{5}{*}{Annealing=10} & Designability $\uparrow$ & $0.528 \scriptstyle { \pm 0.058 }$ & $\textbf{0.788} \scriptstyle { \pm 0.049 }$ & $0.404 \scriptstyle { \pm 0.058 }$ & $0.780 \scriptstyle { \pm 0.051 }$ & $0.716 \scriptstyle { \pm 0.054 }$ & $\underline{\textbf{0.896}} \scriptstyle { \pm 0.038 }$ \\
    & Avg. scRMSD (\AA)$\downarrow$ & $3.721 \scriptstyle { \pm 0.467 }$ & $1.963 \scriptstyle { \pm 0.288 }$ & $3.594 \scriptstyle { \pm 0.378 }$ & $\textbf{1.737} \scriptstyle { \pm 0.136 }$ & $1.972 \scriptstyle { \pm 0.193 }$ & $\underline{\textbf{1.398}} \scriptstyle { \pm 0.124 }$ \\
    & Diversity (Average-TM) $\downarrow$ & $\underline{\textbf{0.338}}$ & $\textbf{0.356}$ & $0.359$ & $0.366$ & $0.362$ & $0.389$ \\
    & \# Designable Clusters $\uparrow$ & $109$ & $\underline{\textbf{123}}$ & $68$ & $\textbf{120}$ & $97$ & $91$ \\
    & Novelty (Max-TM) $\downarrow$ & $\underline{\textbf{0.752}}$ & $0.776$ &  $0.764$ & $\textbf{0.762}$ &  $0.776$ & $0.785$ \\ \midrule
    \multirow{5}{*}{No Annealing} & Designability $\uparrow$ & $0.008 \scriptstyle { \pm 0.011 }$ & $0.028 \scriptstyle { \pm 0.020 }$ & $0.144 \scriptstyle { \pm 0.042 }$ & $\textbf{0.624} \scriptstyle { \pm 0.058 }$ & $0.424 \scriptstyle { \pm 0.056 }$ & $\underline{\textbf{0.800}} \scriptstyle { \pm 0.048 }$ \\
    & Avg. scRMSD (\AA) $\downarrow$ & $12.886 \scriptstyle { \pm 0.388 }$ & $9.643 \scriptstyle { \pm 0.473 }$ & $5.767 \scriptstyle { \pm 0.437 }$ & $\textbf{2.508} \scriptstyle { \pm 0.273 }$ & $3.306 \scriptstyle { \pm 0.349 }$ & $\underline{\textbf{1.698}} \scriptstyle { \pm 0.163 }$ \\
    & Diversity (Average-TM) $\downarrow$ & $\underline{\textbf{0.049}}$ & $\textbf{0.071}$ & $0.291$ & $0.357$ & $0.376$ & $0.368$ \\
    & \# Designable Clusters $\uparrow$ & $2$ & $7$ & $33$ & $\underline{\textbf{111}}$ & $63$ & $\textbf{103}$ \\
    & Novelty (Max-TM) $\downarrow$ & $0.767$ & $\underline{\textbf{0.690}}$ & $\textbf{0.738}$ & $0.748$ & $0.753$ & $0.777$ \\ \bottomrule
\end{tabular}

\caption{ \textbf{Comparison of Rectified models under varying annealing schedules, function evaluations (NFE) and choices of training
coupling for QFlow. “Yes/No” indicates whether inference-time annealing was used. (if Yes, a value of 10 was used for annealing)} We also report the corresponding performance
of the (model, inference setting) pair used to generate each training coupling to give a reference. Many of the observations from studying FoldFlow also carry over, such as the tradeoff in diversity, novelty against designability in whether to use inference scaling when generating the coupling or when generating samples with the rectified model.}
\label{tab:dataimpactreqflow}
\vspace{-1.5em}
\end{table*}

\subsubsection{The use of inverted Examples}
The results of Section \ref{section:dataresults} suggest choosing a coupling whose samples are pareto optimal with respect to designability and diversity. A natural idea used with success for images \cite{zhu2025analyzingmitigatingmodelcollapse, kim2024simplereflowimprovedtechniques}, is to invert groundtruth PDB samples via integrating the base model's flow matching ODE backwards to obtain paired noise samples to the groundtruth PDB. Because it is difficult to match the distribution of the PDB while retaining diversity, real PDB samples represent a desirable trade-off between designability and diversity \cite{geffner2025proteinascalingflowbasedprotein}. We implement this coupling for both ReQFlow \cite{yue2025reqflowrectifiedquaternionflow} and FoldFlow using 100 and 50 steps to invert the coupling respectively, but found that training on the inverted coupling in both cases destroyed model performance, with both models reduced to sub-3\% designability. We hypothesize that this performance collapse is due to the paired noise samples obtained by reverse integration being highly non-Gaussian, and that training on non-Gaussian latents leads to poor performance. Inspired by \cite{bodin2025linearcombinationslatentsgenerative}, we run the Kolmogorov-Smirnov test on the Gaussian latents (treating latents as collections of Gaussian samples) and plot a histogram of $p$-values in Figure \ref{fig:gaussianitytest-pvalues}, where the non-Gaussian nature of the latents is apparent. We subsequently perform an ablation study in Table \ref{tab:latentsablation} which shows that even relatively small perturbations of the fine-tuning latent distribution can lead to severe degradations of performance. As such, the results of both experiments support our hypothesis. 

Varying the inference annealing coefficient, noise injection during backward integration and normalizing the latents did not remedy the issue of non-gaussian latents, with each attempt being either unable to preserve fine-tuned model performance, or to lead to straight flows. Our experiments highlight that not all techniques and design choices for applying ReFlow to images carry over straightforwardly to the protein domain. Nevertheless, we consider the use of the inverted coupling to be a promising future direction for ReFlow in proteins. 

\begin{figure}[t]
    \centering
    \includegraphics[width=0.8\linewidth]{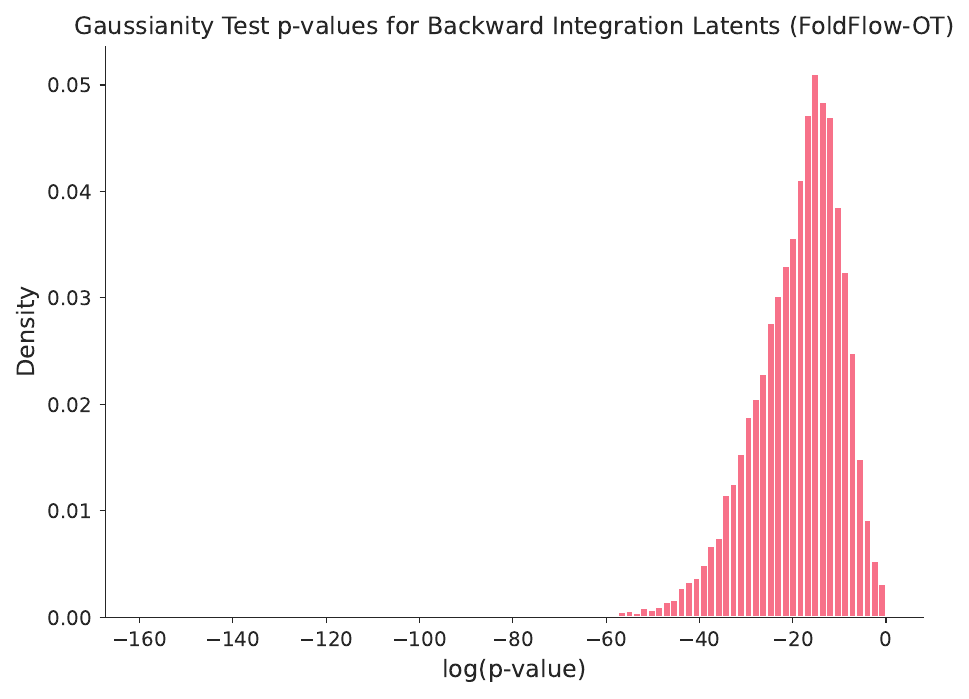}
    \caption{\textbf{FoldFlow backwards integration produces non-Gaussian latents, violating model assumptions.} Distribution of $p$-values (log scale) from Kolmogorov-Smirnov tests shows that 96\% of latents generated by backwards integration have $p < 0.005$, while synthetic centered Gaussian latents achieve high $p$-values (minimum 0.029 across 10,000 samples).}
    \label{fig:gaussianitytest-pvalues}
    \vspace{-1.5em}
\end{figure}

\begin{table}[h]
\centering
\begin{tabular}{lc}
\toprule
\textbf{Finetuning  Coupling} & \textbf{Designability  $\uparrow$} \\
\midrule
\rowcolor{gray!20} 
None & 0.816 \\
Inverted samples & 0.000 \\
IID PDB samples to $\mathcal{N}(0,10^2 I)$ & 0.632 \\
IID PDB samples to $\mathcal{N}(0,8^2 I)$ & 0.432 \\
\bottomrule
\end{tabular}
\caption{\textbf{Inverted coupling causes performance collapse due to latent distribution mismatch.} Designability comparison across FoldFlow-OT coupling variants shows inverted coupling training reduces performance to zero, while PDB resampling maintains comparable performance (0.632). Testing with 0.8-rescaled noise confirms that deviation from the training noise distribution $\mathcal{N}(0,10^2I)$ drives this degradation.}
\vspace{-0.5em}
\label{tab:latentsablation}
\end{table}

\subsection{Improving Training of ReFlow for Proteins}\label{section:improvingreflow}

An aspect of ReFlow fine-tuning on protein models not present in previous areas where ReFlow has been applied is that of structural losses. Previous works on training protein generative models \cite{yim2023se, bose2024se3stochasticflowmatchingprotein, yue2025reqflowrectifiedquaternionflow} have consistently opted to include structural losses during training to penalize physically unrealistic generations, despite deviating from ReFlow theory. While they were found to be useful when training FoldFlow, we show here by example that small details in these losses can have significant effects on the performance of the rectified model:

For FoldFlow in particular, the groundtruth protein backbone is featurized in terms of the 4 heavy atoms in each residue (the $C_\alpha, C, N, O$ atoms) with the $O$ atom given a rotational degree of freedom $\phi$ over idealized alanine coordinates \cite{bose2024se3stochasticflowmatchingprotein}. FoldFlow trains a small MLP head connected to the IPA Network to predict $\phi$, and in Section \ref{section:DataCuration} these values are stored as part of the coupling and used to fine-tune the model during rectification. 

In this section we compare the effect of (1) discarding the structural losses and (2) replacing FoldFlow generated values of $\phi$ in generated couplings with 0 vector. We follow the experimental protocol of Section \ref{section:DataCuration} except for this training detail, and work with both the annealed and semi-annealed datasets, reporting results in Table \ref{tab:varystructuralloss}.

\begin{table}[t]
    \centering
    \scriptsize
    \resizebox{\columnwidth}{!}{
\begin{tabular}{llcccc}
\toprule
\textbf{Dataset} & \textbf{Model} & \textbf{Designability $\uparrow$} & \textbf{scRMSD (\AA) $\downarrow$} & \textbf{Diversity $\downarrow$} & \textbf{Novelty $\downarrow$} \\
\midrule
\multirow{3}{*}{Annealing=10} 
  & Basic ReFlow       & $0.664{\scriptstyle \pm 0.046}$ & $3.156{\scriptstyle \pm 0.381}$ & $0.442$ & $0.822$ \\
  & No Structural Loss & $0.632{\scriptstyle \pm 0.047}$ & $3.720{\scriptstyle \pm 0.427}$ & $0.446$ &  $0.842$       \\
  & Zero-Phi           & $\pmb{0.804{\scriptstyle \pm 0.044}}$ & $\pmb{1.728{\scriptstyle \pm 0.164}}$ & $\pmb{0.424}$ & $\pmb{0.818}$ \\
\midrule
\multirow{3}{*}{Annealing=3} 
  & Basic ReFlow       & $0.492{\scriptstyle \pm 0.049}$ & $4.432{\scriptstyle \pm 0.445}$ & $0.379$ & $\pmb{0.770}$ \\
  & No Structural Loss & $0.544{\scriptstyle \pm 0.047}$ & $4.575{\scriptstyle \pm 0.480}$ & $0.387$ &  $0.800$       \\
  & Zero-Phi           & $\pmb{0.592{\scriptstyle \pm 0.054}}$ & $\pmb{2.945{\scriptstyle \pm 0.371}}$ & $\pmb{0.364}$ &   $0.775$      \\
\bottomrule
\end{tabular}
}
\caption{\textbf{Changing structure losses in ReFlow coupling significantly improves designability without sacrificing diversity.} Performance comparison across structural loss configurations shows that while removing structural loss degrades novelty among designable samples, the $\phi$ replacement strategy provides substantial designability gains without sacrificing novelty or diversity.}
    \label{tab:varystructuralloss}
    \vspace{-1.5em}
\end{table}

While removing structural loss degrades novelty, a surprising result was that avoiding supplying useful $\phi$ information during ReFlow significantly improved model performance, matching or exceeding the performance of the original model at 50 NFE. There are many possible explanations for this observation, ranging from imperfections in model training of the $\phi$ network affecting the rectification procedure, to difficulties in joint numerical optimization of angles, coordinates and the complex and highly specialized structural loss. This illustrates that many domain specific considerations not present in other areas must be made when applying ReFlow to proteins.

We also note that ReQFlow implements the structural loss somewhat differently to FoldFlow and additionally avoids the $\phi$-prediction head, opting to impute the oxygen position using planar geometry. This highlights nuances in the design choices of different models. 

\subsection{Inference Time Settings for ReFlow}\label{section:discretization}

\paragraph{Inference Parameter Selection.} Tables \ref{tab:dataimpact} and \ref{tab:dataimpactreqflow} demonstrate that model behavior changes significantly after rectification for fixed inference parameters, indicating that rectified models require different inference settings than their base counterparts.

\paragraph{Discretization Schemes.} While recent frame-based protein models \citep{bose2024se3stochasticflowmatchingprotein, yue2025reqflowrectifiedquaternionflow} typically employ uniform discretization on $[0,1]$ for ODE integration, we find this approach suboptimal. We systematically evaluate discretization families in Table \ref{tab:discretization_comparison}, categorizing schemes by their density distribution: Noise-Focused (finer near $t=0$), Data-Focused (finer near $t=1$), or Edge-Focused (finer at both endpoints). Representative schemes include the exponential schedule from \citet{geffner2025proteinascalingflowbasedprotein} for Data-Focused, the sigmoid scheme from \citet{kim2024simplereflowimprovedtechniques} for Edge-Focused, and our proposed schedules for Noise-Focused approaches (detailed in Appendix \ref{section:experimentals}).

Noise-Focused discretization achieves superior performance for both QFlow and FoldFlow-OT, with particularly substantial improvements for FoldFlow-OT. This advantage extends to rectified models, where improved discretization schemes enhance overall performance.

The success of Noise-Focused discretization aligns with the inference dynamics of frame-based protein models. As illustrated in Figure \ref{fig:FoldFlowOTDynamics}, FoldFlow-OT exhibits significant trajectory curvature in the Euclidean component near $t=0$, while the $SO(3)$ velocity component shows high magnitude at initialization due to inference annealing. These characteristics complicate ODE dynamics near the noise regime, making finer discretization in this region crucial for reducing integration error.

These findings highlight domain-specific optimization requirements, as \citet{geffner2025proteinascalingflowbasedprotein} report optimal performance with Data-Focused schedules for their protein model. However, their approach does not employ rotation representations or rotation inference annealing, suggesting that differences in discretization preferences may stem from these architectural choices.

\begin{figure*}
    \centering
    \includegraphics[width=1.0\linewidth]{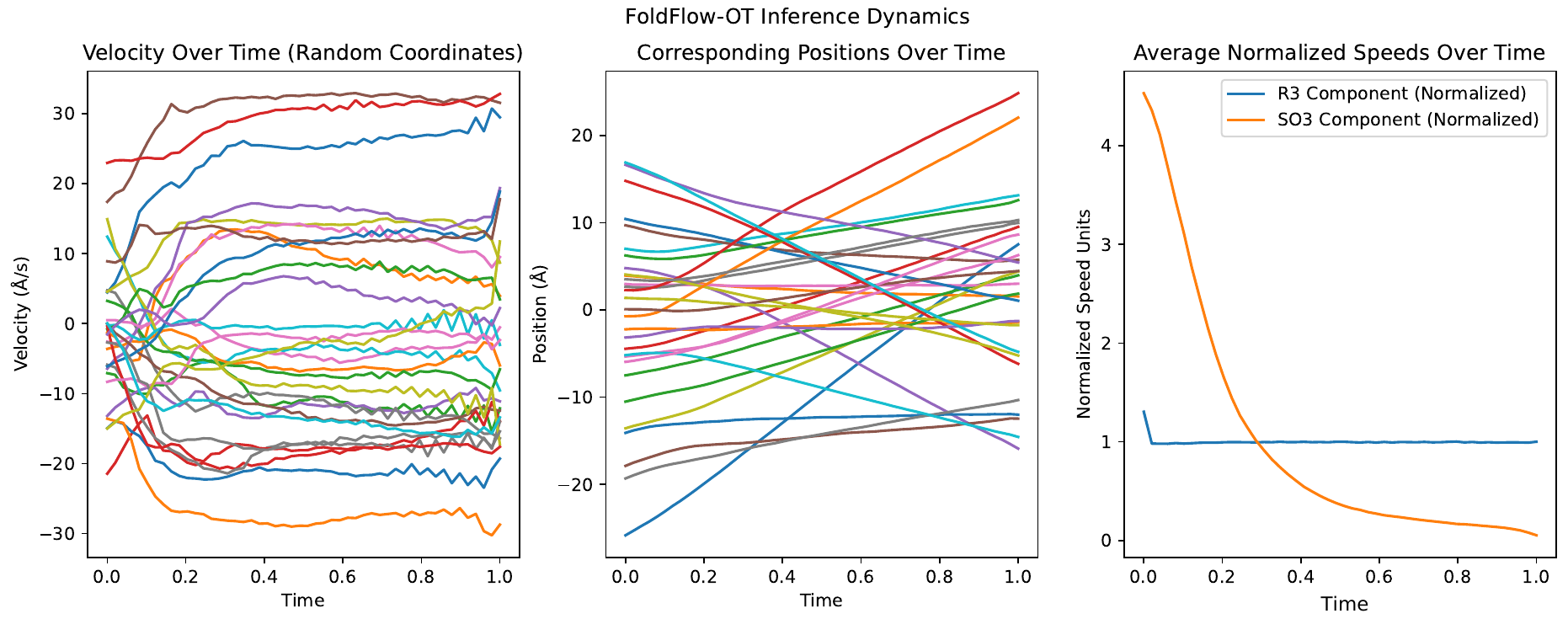}
    \vspace{-2em}
     \caption{\textbf{FoldFlow-OT exhibits complex dynamics near $t=0$ with significant trajectory curvature and rapid frame changes.} Visualization of velocities and positions from 30 random Euclidean coordinates across 50 proteins of varying lengths (100-300 residues, $c=10$ inference annealing) shows substantial curvature in coordinate trajectories near the noise regime, while frame component velocity magnitudes decay significantly after $t\approx 0.3$.}
    \label{fig:FoldFlowOTDynamics}
\end{figure*}

\begin{table}[h]
\centering
\scriptsize
\resizebox{\columnwidth}{!}{
\begin{tabular}{l|cc|cc}
\toprule
\textbf{Discretization Type} & \multicolumn{2}{c|}{\textbf{FoldFlow-OT (15 NFE)}} & \multicolumn{2}{c}{\textbf{QFlow (20 NFE)}} \\
 & \textbf{Designability} $\uparrow$ & \textbf{scRMSD} (\AA) $\downarrow$ & \textbf{Designability}  $\uparrow$ & \textbf{scRMSD} (\AA) $\downarrow$\\
\midrule
Uniform      & $0.076\scriptstyle{\pm 0.031}$ & $7.757\scriptstyle{\pm 0.496}$ & $0.528\scriptstyle{\pm 0.058}$ & $3.721\scriptstyle{\pm 0.467}$ \\
Data-Focused & $0.172\scriptstyle{\pm 0.038}$ & $9.625\scriptstyle{\pm 0.506}$ & $0.276\scriptstyle{\pm 0.050}$ & $6.491\scriptstyle{\pm 0.520}$ \\
Noise-Focused  & $\pmb{0.584\scriptstyle{\pm 0.053}}$ & $\pmb{3.088\scriptstyle{\pm 0.376}}$ & $\pmb{0.584\scriptstyle{\pm 0.056}}$ & $\pmb{2.758\scriptstyle{\pm 0.316}}$ \\
Edge-Focused  & $0.468\scriptstyle{\pm 0.050}$ & $5.510\scriptstyle{\pm 0.516}$ & $0.504\scriptstyle{\pm 0.060}$ & $3.644\scriptstyle{\pm 0.429}$ \\
\bottomrule
\end{tabular}
}
    \caption{\textbf{Noise-Focused discretization outperforms other schemes for frame-based protein models.} Effect of discretization types on designability and scRMSD for FoldFlow-OT and QFlow shows superior performance of Noise-Focused approaches, which are better tailored to the inference dynamics near $t=0$ where trajectory curvature is most significant.}
    \label{tab:discretization_comparison}
    \vspace{-1.5em}
\end{table}

\subsection{When to use ReFlow?}

Given the computational overhead of ReFlow compared to standard fine-tuning approaches, we investigate when ReFlow provides meaningful advantages over simpler techniques such as LoRA fine-tuning \citep{geffner2025proteinascalingflowbasedprotein}. This is as ReFlow requires both ODE integration for each training sample and model fine-tuning on self-generated data, making it considerably more expensive than direct fine-tuning methods.
We compare ReFlow against a control condition using \textit{scrambled} rectified coupling, where noise latents in each noise-data pair are replaced with independent samples from the prior distribution. This maintains identical training conditions while removing the benefits of trajectory rectification. We evaluate both FoldFlow-OT and QFlow models using Annealing=10 datasets, with full experimental details provided in Appendix \ref{section:experimentals}.
Results in Table \ref{tab:resampling_ablation} reveal model-dependent benefits: QFlow shows significant improvement from ReFlow over continued training, while FoldFlow-OT demonstrates no meaningful advantage. We hypothesize that ReFlow's effectiveness diminishes when modeling less diverse data distributions. In the limiting case of flows to point masses, trajectories are already straight and require no rectification.

This hypothesis is supported by the distributional analysis in Figures \ref{fig:secondary structure QFlow} and \ref{fig:secondary-structure-labels}, which shows that FoldFlow-OT's coupling is relatively unimodal and helix-dominated compared to QFlow's more diverse coupling. These findings suggest that ReFlow is most beneficial for highly multimodal distributions, though further investigation is needed to establish definitive guidelines for its application.

\begin{table}[h]
    \centering
    \scriptsize
\resizebox{\columnwidth}{!}{
\begin{tabular}{cc|cccc}
\toprule
\multicolumn{2}{c|}{\textbf{Experiment Configuration}} & \multicolumn{4}{c}{\textbf{Performance Metrics}} \\
\textbf{Model} & \textbf{NFEs} &  \textbf{Designability} $\uparrow$ & \textbf{Avg. scRMSD} (\AA) $\downarrow$ & \textbf{Diversity} $\downarrow$ & \textbf{Novelty} $\downarrow$ \\
\midrule
FoldFlow-Scrambled & 15 & $0.888\scriptstyle{\pm 0.038}$& $1.535\scriptstyle{\pm 0.174}$& $0.425$ & $0.824$\\
FoldFlow-ReFlow & 15 & $0.844\scriptstyle{\pm 0.043}$ & $1.748\scriptstyle{\pm0.210}$ & $0.404$ & $0.811$\\
\midrule 
QFlow-Scrambled & 20 & $0.648\scriptstyle{\pm 0.058}$ & $2.711\scriptstyle{\pm 0.330}$& $0.361$ & $0.739$\\
QFlow-ReFlow & 20 & $0.780\scriptstyle{\pm 0.051}$ & $1.737\scriptstyle{\pm 0.136}$ & $0.366$ & $0.762$\\ 
\bottomrule
\end{tabular}
}
    \caption{\textbf{ReFlow provides greater benefits for multimodal distributions.} Comparison of model pairs finetuned on model samples (scrambled coupling) versus ReFlow under identical settings shows that while ReFlow training significantly increased QFlow designability at low NFE, no benefits were observed for FoldFlow-OT, potentially due to differences in data distribution multimodality.}
    \label{tab:resampling_ablation}
\end{table}

\section{Discussion}

Our research shows that ReFlow's mathematics can readily be extended to manifold data and thereby to frame-based protein backbone design; however, the inference acceleration without quality or diversity loss observed in computer vision applications requires significant domain-specific optimization. Our experiments highlight two main challenges.
First, we observe that several common optimizations of ReFlow for images, such as the use of inverted coupling strategies and temporal discretization scheme, counterintuitively proved detrimental for ReFlow in frame-based backbone generation. 
Conversely, substantial performance improvements were achieved through careful optimization of domain-specific design choices, such as inference annealing schedules and protein-tailored components in the loss function.
These findings underscore that successful adaptation of generative models across domains requires considering domain-specific constraints rather than direct transfer of existing techniques. 
A second significant challenge concerns the evaluation of generated protein diversity. As discussed in Sections \ref{section:DataCuration} and \ref{section:discretization}, achieving appropriate balance between sample quality and diversity is critical for practical protein design applications. However, standard diversity metrics such as average pairwise TM-score suffer from limited interpretability, making it difficult to assess the practical significance of differences between model performances. In contrast, biologically-informed metrics such as secondary structure distributions (Figure \ref{fig:secondary-structure-labels}) provide more interpretable insights into the functional diversity of generated protein ensembles.
These observations suggest two promising directions for future research: systematic investigation of how diffusion and flow acceleration techniques, including consistency models \citep{song2023consistencymodels} and flow map matching \citep{boffi2024flowmapmatching}, can be effectively adapted to protein structure generation, and development of improved diversity metrics that better capture biologically relevant aspects of protein structural variation \citep{geffner2025proteinascalingflowbasedprotein}.

\section*{Acknowledgements}

The authors would like to thank Francisco Vargas for his
insightful discussions that contributed to this work. 

\section*{Impact Statement}
This paper aims to study and advance the efficiency of current frame based methods for protein backbone design. No direct societal consequences are forseen by this work.

\bibliography{example_paper}
\bibliographystyle{icml2025}

\newpage
\appendix
\onecolumn

\section{Proof of Prop \ref{prop:reudcetransport costs}}\label{section:proof}

Define for convenience $X_t=I(X_0,X_1,t)$ and $\dot X_t=\partial_tI(X_0,X_1,t)$ where $\partial_t$ represents the derivative with respect to $t$. Let $u(x,t)$ denote the rectified velocity field from minimizing $\mathcal{L}$ in Algorithm \ref{alg:ReFlow}. We have: 
\begin{equation}
\begin{aligned}
\mathbb{E}[d_g(X_0, X_1)]&=\mathbb{E}_{(X_0,X_1)}\left[\int_0^1 \norm[X_t]{\partial_t I(X_0,X_1,t)} dt\right]\ \text{ \hspace{1cm}\color{gray} (as $I(X_0,X_1,\cdot)$ traces a geodesic)}\\
&= \int_0^1 \mathbb{E}_{(X_0,X_1)}\left[\|\dot X_t\|_{X_t}\right]dt \text{ \hspace{1.7cm}\color{gray} (Fubini's theorem)}\\
&=\int_0^1 \mathbb{E}_{X_t,X_1}\left[\|\dot X_t\|_{X_t}\right]dt\ \text{ \hspace{1.8cm}\color{gray}(as $\dot X_t$ at $X_t$ depends only on $X_t$, $X_1$ and $t$)}\\
&=\int_0^1 \mathbb{E}_{X_t}\left[\mathbb{E}_{X_1}[\|\dot X_t\|_{X_t}|X_t]\right]dt\ \text{\hspace{1cm}\color{gray} (tower property)}\\
&\geq \int _0^1 \mathbb{E}_{X_t}\left[\|\mathbb{E}_{X_1}[\dot X_t|X_t]\|_{X_t}\right]dt \ \text{ \hspace{1cm}\color{gray}(Jensen's inequality)}\\
&= \int_0^1 \mathbb{E}_{X_t}\left[\|u(X_t,t)\|_{X_t}\right]dt \ \text{\hspace{0.5cm}\color{gray} (the rectified velocity field, $u(X_t,t)=\mathbb{E}\left[\dot X_t|X_t\right]$ c.f \cite{wu2025riemannianneuralgeodesicinterpolant})}\\
&= \int_0^1 \mathbb{E}_{Z_t}\left[\|u(Z_t,t)\|_{Z_t}\right]dt\ \hspace{1.8cm}\color{gray}\text{(as }Z_t \text{ and } X_t \text{ have same marginals)} \\
&=\mathbb{E}_{(Z_t)_{0\leq t \leq 1}}\left[\int_0^1 \|u(Z_t,t)\|_{Z_t}dt\right]\\
&\geq \mathbb{E}_{(Z_0,Z_1)}\left[d_g(Z_0,Z_1)\right] 
\end{aligned}
\end{equation}
where the last inequality follows from the fact that $\int_0^1 \|u(Z_t,t)\|_{Z_t}dt$ corresponds to the length of the ODE-induced path from $Z_0$ to $Z_1$ and thus is an upper bound on $d_g(Z_0,Z_1)$. $\square$

\section{Experimental Details}\label{section:experimentals}

\subsection{Evaluation Metrics}

We follow exactly the evaluation protocol of \cite{bose2024se3stochasticflowmatchingprotein}, which was also used in \cite{yue2025reqflowrectifiedquaternionflow, geffner2025proteinascalingflowbasedprotein}. We generate 50 backbones each of sizes [100,150,200,250,300], and evaluate our metrics (Designability, scRMSD, Diversity, Novelty, Number of Clusters) on the generated backbones. We give a description of the metrics used:

\paragraph{Designability and scRMSD:} We use this metric to evaluate whether a protein backbone can be formed by
folding an amino acid chain. For each backbone, we generate 8 sequences with ProteinMPNN \cite{dauparas2022robust} at
temperature 0.1, and predict their corresponding structures using ESMFold \cite{lin2022language}. Then we compute
the minimum RMSD (known as scRMSD) between the predicted structures and the backbone
sampled by the model. The designability score is the fraction
of samples satisfying scRMSD being less than $2Å$. We use a corrected version of the aligned-RMSD implementation in \cite{yim2023se}, leading to slight deviations from previously reported values.

\paragraph{Diversity:} This metric quantifies the diversity of the generated backbones. This involves calculating
the average pairwise structural similarity among designable samples, broken down by protein length. Specifically, for each length $L$ under consideration, let $S_L$ be the set of designable structures. We compute $TM(b_i,b_j)$ for all distinct pairs $(b_i, b_j)$ within $S_L$. The mean of these TM-scores represents the diversity for length L. The final diversity score is the average of these means across all tested lengths L. Since TM-scores closer to 1 indicate higher similarity, superior diversity is reflected by lower values of this aggregated score.

\paragraph{Novelty:} We evaluate the structural novelty by finding the maximum TM-score between a generated
structure and any structure in the training set for FoldFlow, reproduced following the instructions of \citet{yim2023se} and containing 23474 entries, using Foldseek \cite{article}. A lower resulting
maximum TM-score signifies a more novel structure. The Foldseek command used for novelty was taken from \citet{geffner2025proteinascalingflowbasedprotein}. Novelty was only evaluated on designable structures.

\paragraph{Number of Designable Clusters:} We cluster the designable backbones based on a TM-score threshold of 0.5 using Foldseek and report the number of different clusters. This metric balances designability and diversity, measuring the number of "distinct" designable proteins that the model generates. The command used to perform the clustering was taken from \citet{geffner2025proteinascalingflowbasedprotein}.

\subsection{Training Details}
\paragraph{Data Curation Experiments:} We generate the FoldFlow-OT coupling using 50 integration steps and default hyperparameters from the repository with the exception of $SO(3)$ inference scaling which was varied. We did not find additional integration steps helpful for sample quality for FoldFlow. Based on the results of \cite{yue2025reqflowrectifiedquaternionflow}, we generate QFlow samples with 100 integration steps. A single coupling dataset made for each of the \texttt{No Annealing}, \texttt{Annealing=3}, \texttt{Annealing=10} settings and reused for all experiments. We sample $t\sim U[0,1]$.

We retain all hyperparameters from FoldFlow's training codebase. We opt to use the Axis-Angle rotation loss which FoldFlow also used for training, as we found it to be marginally beneficial compared to the $L^2$ loss. We perform rectification training for $20,000$ batches for FoldFlow and 3 epochs for QFlow (approximately $24, 000$ batches). Note that a fair comparison of QFlow and FoldFlow is beyond the scope of this work, and these parameters are only meant to serve as reasonable values for each base model.

\paragraph{Improved Training Experiments:} We use the same training parameters as FoldFlow-OT except for the structural losses.

\paragraph{Discretization:} We supply the choices of discretization used here. We present the non-uniform discretizations used in Table \ref{tab:discretization_schedules}. We generated the Data and Edge Focused schedules using the formulas given in \citet{geffner2025proteinascalingflowbasedprotein} and \citet{kim2024simplereflowimprovedtechniques} respectively. For the Edge-Focused Sigmoid schedule, we pick $\kappa=5$ as a reasonable
value with similar spacing properties to our Noise-Focused schedules. Our Noise-Focused schedule for FoldFlow was handcrafted based on observed inference dynamics in Fig.  \ref{fig:FoldFlowOTDynamics} and testing on a small number of examples. While it works well, there is likely still room for improvement. The Noise-Focused schedule for QFlow was a time-reversal of the Data-Focused Schedule. We report the designabilities averaged over 50 samples each of lengths [100,150,200,250,300]. 

\begin{table*}[h]
\centering
\scriptsize
\begin{tabular}{lcc}
\toprule
\textbf{Discretization Class} & \textbf{FoldFlow (15 NFE)} & \textbf{ReQFlow (20 NFE)} \\
\midrule
\makecell[l]{\textbf{Data-Focused}} &
\makecell[l]{[0.010, 0.040, 0.083, 0.129, 0.180,\\ 0.235, 0.295, 0.360, 0.430,\\ 0.507, 0.590, 0.680, 0.778,\\ 0.885, 1.000]} &
\makecell[l]{[0.010, 0.121, 0.229, 0.325, 0.410,\\ 0.486, 0.554, 0.615, 0.669,\\ 0.717, 0.760, 0.798, 0.832,\\ 0.862, 0.889, 0.914, 0.935,\\ 0.954, 0.971, 0.986, 1.000]} \\
\midrule
\makecell[l]{\textbf{Edge-Focused}} &
\makecell[l]{[0.010, 0.034, 0.080, 0.138, 0.211,\\ 0.298, 0.396, 0.500, 0.604,\\ 0.702, 0.789, 0.862, 0.920,\\ 0.966, 1.000]} &
\makecell[l]{[0.010, 0.023, 0.051, 0.085, 0.126,\\ 0.173, 0.228, 0.289, 0.356,\\ 0.427, 0.500, 0.573, 0.644,\\ 0.711, 0.772, 0.827, 0.874,\\ 0.915, 0.949, 0.977, 1.000]} \\
\midrule
\makecell[l]{\textbf{Noise-Focused}} &
\makecell[l]{[0.010, 0.025, 0.050, 0.075, 0.100,\\ 0.125, 0.150, 0.175, 0.200,\\ 0.225, 0.250, 0.275, 0.460,\\ 0.640, 0.820]} &
\makecell[l]{[0.010, 0.014, 0.029, 0.046, 0.065,\\ 0.086, 0.111, 0.138, 0.168,\\ 0.202, 0.240, 0.283, 0.331,\\ 0.385, 0.446, 0.514, 0.590,\\ 0.675, 0.771, 0.879, 1.000]} \\
\bottomrule
\end{tabular}
\caption{Discretization schedules used in FoldFlow and ReQFlow for different classes of time discretizations.}
\label{tab:discretization_schedules}
\end{table*}

\paragraph{Scrambling Experiments:} For the QFlow pair of models the training procedure is exactly the same as for the Data Experiments. However, we used a fine-tuned rectification procedure for the FoldFlow pair of models that attained the highest designability. In particular, we use improved structural losses practices from Section \ref{section:improvingreflow}, and sample $p(t)\propto 0.1^t$ instead of $p(t)\propto 1$ (i.e prioritizing noisy samples), and sample proteins with length $l$ with probability proportional to $l^{2.5}$ in our dataset. 

\section{Straightness in Frame-based Flow models}

As Frame-based Flow models have a significant Euclidean component, we found it helpful to track for diagnostic purposes an analogous \textit{straightness} metric for the Euclidean component of $SE(3)^N$ as the one defined on $\mathbb{R}^N$ in \cite{liu2022flowstraightfastlearning}. In particular, we monitor the \textit{straightness} of the Euclidean part of $SE(3)^N$, defined as  
\begin{equation}
    STN_C=\frac{1}{N_r}\int_0^1 \|X_1^C-X_0^C-v^C(X_t,t)\|^2dt
\end{equation}
where the $C$ denotes the \textbf{C}oordinate part of the protein representation and $N_r$ is the number of residues. Note that a lower value of $STN_C$ indicates straighter paths, and a flow with geodesic trajectories in $SE(3)^N$ (under the Riemannian metric of \citep{bose2024se3stochasticflowmatchingprotein}) will have $STN_C$ necessarily equal 0 (i.e the Euclidean component of the trajectory is a straight line). 

We compute the average straightness for FoldFlow-OT and a rectified model in Table \ref{tab:straightness_results} on 50 samples each of lengths [100,150,200,250,300],  and found that applying ReFlow in $SE(3)^N$ did significantly reduce $STN_C$, and that $STN_C$ was not significantly affected by inference annealing. This is further confirmed qualitatively in Fig \ref{fig:trajectoriesReFlowed}, where the rectified FoldFlow-OT model exhibits qualitatively straighter trajectories compared to the base model in Fig \ref{fig:FoldFlowOTDynamics}.
\begin{table}[h!]
\centering
\resizebox{0.5\textwidth}{!}{
\begin{tabular}{@{}lccc@{}}
\toprule
\textbf{Model} & \textbf{Use Inference Annealing} &  \textbf{$STN_C$} \\
\midrule
FoldFlow-OT   & No        & $45.627\scriptstyle{\pm0.663}$ \\
FoldFlow-OT & Yes        & $34.913\scriptstyle{\pm0.535}$ \\
FoldFlow-OT-Reflowed & No   & $5.642\scriptstyle{\pm0.017}$  \\
FoldFlow-OT-Reflowed & Yes & $6.236\scriptstyle{\pm0.033}$  \\
\bottomrule
\end{tabular}
}
\caption{\textbf{Comparison of  $STN_C$ straightness metrics across models}, with and without inference-time scaling. Results are reported with confidence intervals and computed using 50 uniformly spaced integration steps. Rectification leads to straighter paths in the $\mathbb{R}^3$ components. Inference Scaling does not affect straightness in the Euclidean component.}
\label{tab:straightness_results}
\end{table}

\begin{figure}
    \centering
    \includegraphics[width=1.0\linewidth]{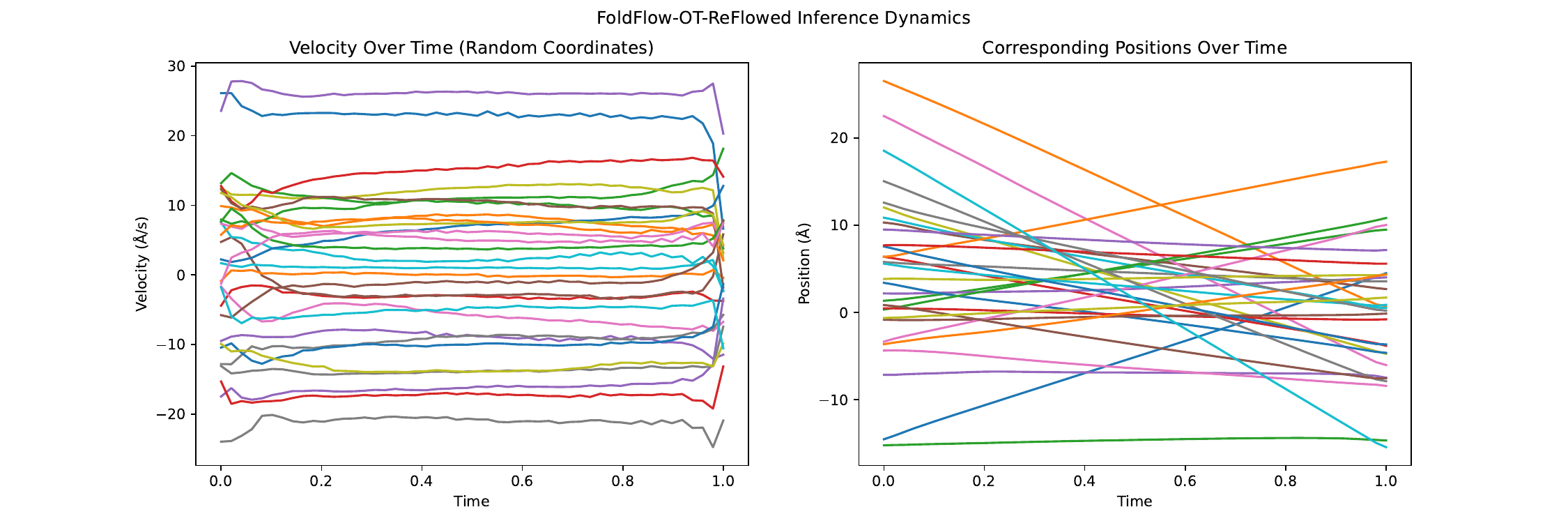}
    \caption{We repeat the visualization for Fig.\ref{fig:FoldFlowOTDynamics} for our rectified model. The position coordinate trajectories of our rectified model shows significantly less curvature especially near $t=0$ and look almost indistinguishable from straight lines.}
    \label{fig:trajectoriesReFlowed}
\end{figure}

\end{document}